\documentclass[11pt,english]{article}

\usepackage[T1]{fontenc}
\usepackage[latin9]{inputenc}
\usepackage{geometry}
\usepackage{color}
\usepackage{babel}
\usepackage{array}
\usepackage{textcomp}
\usepackage{bm}
\usepackage{multirow}
\usepackage{amsmath}
\usepackage{amsthm}
\usepackage{amssymb}
\usepackage{graphicx}
\usepackage{setspace}
\usepackage[numbers]{natbib}
\usepackage{microtype}
\usepackage[unicode=true,pdfusetitle,
 bookmarks=true,bookmarksnumbered=false,bookmarksopen=false,
 breaklinks=false,pdfborder={0 0 0},pdfborderstyle={},backref=false,colorlinks=true]
 {hyperref}
\hypersetup{
 urlcolor=blue, citecolor=blue, linkcolor=red}

\makeatletter

\providecommand{\tabularnewline}{\\}

\theoremstyle{plain}
\newtheorem{lem}{\protect\lemmaname}
\theoremstyle{plain}
\newtheorem{thm}{\protect\theoremname}
\theoremstyle{plain}
\newtheorem{prop}{\protect\propositionname}
\theoremstyle{remark}
\newtheorem*{rem*}{\protect\remarkname}
\theoremstyle{remark}
\newtheorem*{acknowledgement*}{\protect\acknowledgementname}

\newlength{\abstractwidth}
\g@addto@macro\normalsize{\setlength{\abovedisplayskip}{10pt}\setlength{\belowdisplayskip}{20pt}\setlength{\abovedisplayshortskip}{10pt}\setlength{\belowdisplayshortskip}{20pt}}

\makeatother

\providecommand{\acknowledgementname}{Acknowledgement}
\providecommand{\lemmaname}{Lemma}
\providecommand{\propositionname}{Proposition}
\providecommand{\remarkname}{Remark}
\providecommand{\theoremname}{Theorem}

\begin{document}

\title{\textbf{\Large{}Conditions for local transformations between sets
of quantum states }}

\author{Somshubhro Bandyopadhyay\thanks{Department of Physics, Bose Institute, EN 80, Sector V, Bidhannagar,
Kolkata 700091, India; \textcolor{blue}{som.s.bandyopadhyay@gmail.com},
\textcolor{blue}{som@jcbose.ac.in}} $\hspace{1em}$ Saronath Halder\thanks{Centre for Quantum Optical Technologies, Centre of New Technologies,
University of Warsaw, Banacha 2c, 02-097 Warsaw, Poland; \textcolor{blue}{saronath.halder@gmail.com}} $\hspace{1em}$ Ritabrata Sengupta\thanks{Department of Mathematical Sciences, Indian Institute of Science,
Education, \& Research (IISER) Berhampur, Transit campus, Govt. ITI,
NH 59, Berhampur 760 010, Ganjam, Odisha, India; \textcolor{blue}{rb@iiserbpr.ac.in}}}
\maketitle
\begin{abstract}
We study the problem of transforming a set of pure bipartite states
into another using deterministic LOCC (local operations and classical
communication). Necessary conditions for the existence of such a transformation
are obtained using LOCC constraints on state transformation, entanglement,
and distinguishability. These conditions are shown to be independent
but not sufficient. We discuss their satisfiability and classify all
possible input-output pairs of sets accordingly. We also prove that
strict inclusions hold between LOCC, separable, and positive partial
transpose operations for set transformation problems. 
\end{abstract}

\section{Introduction}

Many of the distinctive features of quantum theory appear in scenarios
that impose constraints on quantum operations. One such scenario is
the so-called \emph{distant lab} paradigm in which quantum operations
on a composite system are restricted to LOCC (local operations and
classical communication), a subset of all quantum operations that
one may perform on the system under consideration. The LOCC setup
considers a quantum system shared between two or more physically separated
observers who can perform any quantum operation on their local subsystems
and communicate via classical channels but cannot exchange quantum
information. Quantum operations that can be realized in this fashion
belong to the LOCC class \citep{Chitambar+LOCC}. Many of the fundamental
questions in quantum information theory, especially those related
to quantum nonlocality \citep{Bell-nonlocality-review,NLWE,Horodeckis+MNLE,Bandyo-2011,Halder+-2019},
resource theories \citep{resource-theories-review}, state discrimination
\citep{NLWE,ben99u,Horodeckis+MNLE,Bandyo-2011,Halder+-2019,Peres-Wootters,Walgate+2000,Bell-indistinguishable,Walgate-Hardy-2002,Watrous-2005,Nathanson-2005,BGK-2011,Yu-Duan-2012,B-IQC-2015,BN-2013,Cosentino-2013,Cosentino-Russo-2014},
and entanglement distillation \citep{entanglement-review}, are studied
within the framework of LOCC. 

One of the goals of quantum information theory is to understand the
strengths and limitations of LOCC. The motivation mainly comes from
the resource theory of entanglement \citep{resource-theories-review,entanglement-review},
which considers quantum operations belonging to the LOCC class as
\emph{free} but those requiring entangled states, as a resource, expensive.
Thus it is necessary to identify which quantum operations can, in
fact, be realized by LOCC and which cannot be. The standard approach
to do so, in general, is a two-step process: first, we specify a quantum
task\textendash a measurement, a unitary operation, or something more
general, such as entanglement transformation or state discrimination,
and then try to work out, often with the help of known LOCC constraints,
whether LOCC can perform this task just as well as global operations,
and if LOCC is sub-optimal, only then do we need to call upon entanglement.

Some of the well-known limitations of LOCC manifest themselves in
quantum state transformations. For example, LOCC cannot convert a
separable state into an entangled state, increase entanglement on
average, or increase the Schmidt rank of a bipartite state. These
properties may be applied to identify quantum operations not realizable
by LOCC. For example, consider the result stating that LOCC implementation
of Bell measurement is impossible \citep{Bell-indistinguishable}\textendash{}
here, the proof simply follows from the demonstration that LOCC implementation
of Bell measurement implies LOCC conversion of a product state into
a Bell state. Likewise, applying the second property, one can show
that LOCC cannot convert a nonmaximally entangled state into a maximally
entangled state belonging to the same state space with unit probability.
Other notable results include \emph{quantum nonlocality without entanglement}
and the proof that LOCC is a strict subset of separable quantum operations
\citep{NLWE}, the existence of incomparable pairs \citep{Nielsen-99},
and the grouping of entangled states into SLOCC (stochastic LOCC)
inequivalence classes \citep{entanglement-review}.

One of the fundamental problems in quantum information theory is quantum
set transformation, which, simply put, deals with the task of transforming
a set of input states into a set of output states. The significance
of this problem or variants thereof lies in the fact that many quantum
information processing tasks can be naturally understood as set transformations,
especially in situations where a quantum channel acts on an ensemble
of input or signal states. The general set transformation problem
can be described as follows. Suppose that $S_{\rho}=\left\{ \rho_{1},\dots,\rho_{n}\right\} $
is a given set of input states and $S_{\sigma}=\left\{ \sigma_{1},\dots,\sigma_{n}\right\} $
is a set of output states, and the goal is to achieve the transformation
$\rho_{i}\rightarrow\sigma_{i}$ for every $i$. This, however, is
possible only when there exists a quantum operation $\bm{\Lambda}$
such that $\bm{\Lambda}\left(\rho_{i}\right)=\sigma_{i}$ for all
$i.$ So the question is: what necessary and sufficient conditions
must be satisfied for the existence of a $\Lambda$ such that the
desired transformation is possible? Note that the scenario we just
described does not have restrictions on quantum operations and has
been studied before but with limited success \citep{Alberti-Uhlman-1980-1,AlbertiUhlman-1980-2,Uhlmann-1985,Chefles-2000,Chefles+}. 

In this paper, we investigate the set transformation problem within
the paradigm of LOCC. The motivation is three-fold. First, LOCC provides
the framework to study set transformations in a distributed setting
which, as discussed earlier, is the natural setup in many quantum
information-theoretic tasks. But this scenario, as far as we know,
has not been systematically studied before (a somewhat closely related
problem was discussed recently in Ref.$\,$\citep{BH-2021}). Second
is the observation that the LOCC problem has certain distinguishing
features that do not occur in the general formulation simply because
they arise from orthogonality of states, quantum entanglement, and
local distinguishability, none of which is significant when there
are no restrictions on quantum operations. Our third motivation stems
from the earlier discussion on the need to understand the strengths
and limitations of LOCC protocols, which largely depends on investigating
how well LOCC protocols perform well-defined quantum tasks, which,
in this case, is quantum set transformation. 

Specifically, we investigate conditions under which a given set of
pure bipartite states can be deterministically converted into another
using LOCC. We obtain necessary conditions by applying LOCC constraints
on state transformations, entanglement, and distinguishability. We
show that they are independent but not sufficient. To demonstrate
independence, we consider a proper subset of the necessary conditions
and then show that for every such subset, there exists a pair of input
and output sets that satisfy the members of the subset under consideration
but violate others that do not belong to it. This leads to a natural
classification of all possible input-output set pairs with distinct
and sometimes overlapping properties. To prove the insufficiency,
we present an input-output set pair that satisfies all the necessary
conditions but for which the desired set transformation is still impossible. 

We also discuss the set transformation problem vis-à-vis other notable
classes of quantum operations, namely, the separable operations (SEP)
and positive partial-transpose operations (PPT). It is well-known
that the following relations hold (see, for example, \citep{Chitambar+LOCC}):
\[
\text{LOCC}\subset\text{SEP}\subset\text{PPT}\subset\text{ALL},
\]
where ALL denotes the set of all possible quantum operations. The
above inclusions are strict. However, for specific problems, they
may not be; for example, the classes are all equally good for distinguishing
two mutually orthogonal pure states because any two orthogonal pure
states can be perfectly distinguished by LOCC \citep{Walgate+2000}.
So the question here is, are the inclusions strict for set transformations?
We will answer this question in the affirmative. 

\section{Transformations between sets of quantum states }

Let $S_{\rho}=\left\{ \rho_{1},\dots,\rho_{n}\right\} $ and $S_{\sigma}=\left\{ \sigma_{1},\dots,\sigma_{n}\right\} $
be two sets of quantum states. Suppose that a quantum system is now
prepared in one of the states chosen from $S_{\rho}$ but we do not
know which state it is. The goal is to find out whether there exists
a \emph{deterministic} quantum transformation\footnote{One may consider a probabilistic formulation in which $\rho_{i}\rightarrow\sigma_{i}$
with some nonzero probability $p_{i}$ for every $i$ \citep{Chefles+}.}, a linear completely positive trace-preserving map $\bm{\Lambda}$
satisfying $\bm{\Lambda}\left(\rho_{i}\right)=\sigma_{i}$ for all
$i$. Henceforth, we will call $S_{\rho}$ the input set, $S_{\sigma}$
the output set, and use the notation $S_{\rho}\rightarrow S_{\sigma}$
to indicate $\bm{\Lambda}\left(\rho_{i}\right)=\sigma_{i}$ for all
$i$. 

Note that if the input states are mutually orthogonal, the problem
is trivial, for one could first learn the identity of the input state
with appropriate measurement and then prepare the output state accordingly.
So the input states, at the very least, cannot be all mutually orthogonal.
The problem, however, mostly remains open, and the answers are known
only for specific cases \citep{Uhlmann-1985,Chefles+}. 

The fidelity $\bm{f}\left(\rho,\sigma\right)$ between two quantum
states $\rho$ and $\sigma$ is defined as
\begin{flalign}
\bm{f}\left(\rho,\sigma\right) & =\text{Tr}\sqrt{\sqrt{\sigma}\rho\sqrt{\sigma}}.\label{fidelity}
\end{flalign}
If $\rho=\left|\psi\right\rangle \left\langle \psi\right|$ and $\sigma=\left|\phi\right\rangle \left\langle \phi\right|$
correspond to pure states, then 
\begin{flalign}
\bm{f}\left(\psi,\phi\right) & =\left|\left\langle \psi\vert\phi\right\rangle \right|.\label{fidelity-pure-states}
\end{flalign}
The following results have been proved.
\begin{lem}
\citep{Uhlmann-1985} \label{Lemma1} If $S_{\rho}\rightarrow S_{\sigma}$
then $\bm{f}\left(\sigma_{i},\sigma_{j}\right)\geq\bm{f}\left(\rho_{i},\rho_{j}\right)$
for all $1\leq i,j\leq n$. 
\end{lem}
Lemma \ref{Lemma1} is a necessary condition that tells us that the
pairwise distinguishability can never increase under a deterministic
quantum operation. This is, of course, intuitively expected. The condition
is also sufficient in the following case \citep{AlbertiUhlman-1980-2}. 
\begin{lem}
\label{Lemma 2} Let $S_{\rho}=\left\{ \rho_{1},\rho_{2}\right\} $
and $S_{\sigma}=\left\{ \sigma_{1},\sigma_{2}\right\} $, where $\rho_{1}$
and $\rho_{2}$ are pure states. Then $S_{\rho}\rightarrow S_{\sigma}$
if and only if $\bm{f}\left(\sigma_{1},\sigma_{2}\right)\geq\bm{f}\left(\rho_{1},\rho_{2}\right)$. 
\end{lem}
Lemma \ref{Lemma 2} can be improved upon for qubit states. Let $\left\Vert \mathcal{O}\right\Vert _{1}=\text{Tr}\sqrt{\mathcal{O}^{\dagger}\mathcal{O}}$
be the trace-norm of an operator $\mathcal{O}$. If $\mathcal{O}$
is Hermitian, then $\left\Vert \mathcal{O}\right\Vert _{1}=\sum_{i}\left|\lambda_{i}\right|$,
where $\lambda_{i}$ are the eigenvalues of $\mathcal{O}$.
\begin{lem}
\citep{AlbertiUhlman-1980-2} \label{Lemma3} Let $S_{\rho}=\left\{ \rho_{1},\rho_{2}\right\} $
and $S_{\sigma}=\left\{ \sigma_{1},\sigma_{2}\right\} $, where $\rho_{1},\rho_{2},\sigma_{1},\sigma_{2}$
are qubit states. Then $S_{\rho}\rightarrow S_{\sigma}$ if and only
if 
\begin{flalign}
\left\Vert p_{1}\sigma_{1}-p_{2}\sigma_{2}\right\Vert _{1} & \leq\left\Vert p_{1}\rho_{1}-p_{2}\rho_{2}\right\Vert _{1}\label{Lemma3-equation}
\end{flalign}
for all probability distributions $\left\{ p_{1},p_{2}\right\} $. 
\end{lem}
Necessary conditions have also been obtained for sets of pure states,
which are also sufficient provided the input states are linearly independent
\citep{Chefles-2000}. The good news, however, pretty much ends here.
Lemma \ref{Lemma 2}, for example, fails to be sufficient if $\rho_{1}$
and $\rho_{2}$ are not pure states \citep{Chefles+}. One can also
construct explicit counterexamples to show that Lemma \ref{Lemma1}
is not sufficient either \citep{Jozsa+1999}. The necessary and sufficient
conditions for transforming a set of pure states to a set of pure
or mixed states have also been obtained \citep{Chefles+}, but they
may not easy to deal with in general. 

\section{LOCC transformations between sets of quantum states }

The LOCC set transformation problem can be formulated by making the
expected changes to the general formulation: the quantum system is
now composite, that is, composed of two or more subsystems, and the
quantum operations now belong to the LOCC class. 

In this paper, we will only consider bipartite quantum systems $\mathcal{H}=\mathcal{H}_{A}\otimes\mathcal{H}_{B}$
of finite dimensions. Let $\mathcal{D}\left(\mathcal{H}\right)$ denote
the set of density operators on $\mathcal{H}$, and let $S_{\varrho}=\left\{ \varrho_{1},\dots,\varrho_{n}\right\} \subset\mathcal{D}\left(\mathcal{H}\right)$
be the set of input states and $S_{\varsigma}=\left\{ \varsigma_{1},\dots,\varsigma_{n}\right\} \subset\mathcal{D}\left(\mathcal{H}\right)$
be the set of output states, where $n\geq2$. So the question here
is the following: Does there exist a \emph{deterministic} LOCC transformation
$\mathcal{\bm{L}}$ such that $\mathcal{\bm{L}}\left(\varrho_{i}\right)=\varsigma_{i}$
for all $i$?\footnote{Probabilistic formulation is also possible similar to entanglement
transformations \citep{Vidal-1999}. } 

The operational meaning here is also quite clear. We assume that two
physically separated observers, Alice and Bob, share a quantum state
chosen from $S_{\varrho}$ but do not know which state they actually
share, but their goal is to produce the correct output state with
certainty using a LOCC protocol. Then for which input-output pairs
of sets they could perform this task? 

Before we get to the problem in detail, let us briefly discuss the
essential differences between the LOCC and the general problem. We
will see that by restricting the class of quantum operations to LOCC,
we encounter situations that do not arise otherwise. 

\subsection*{Orthogonality of states}

Recall that, in the general case, the input states cannot be all pairwise
orthogonal because the problem becomes trivial due to the \textquotedblleft identify
and prepare (IP)\textquotedblright{} strategy. In the LOCC scenario,
as we will now explain, the input states can be mutually orthogonal
and yet give rise to nontrivial situations as the IP strategy does
not always work. 

Let us recall some basic concepts and results from LOCC distinguishability
of quantum states which we will also need later in this paper. One
of the fundamental results is that a set of bipartite or multipartite
orthogonal states cannot always be perfectly distinguished (i.e.,
without error) by LOCC \citep{NLWE,Walgate+2000}, although they can
be, as we know, by a joint measurement on the whole system. Specifically,
if a given set of states is perfectly distinguishable by LOCC, then
it means that the unknown state can be identified without error using
some LOCC protocol. Now, this can happen either with certainty or
with some nonzero probability. We will call a set of orthogonal states
locally distinguishable if the members can be perfectly distinguished
by LOCC with certainty, else locally indistinguishable. Note that
a locally indistinguishable set comes in two varieties. The first
is where a LOCC protocol cannot distinguish the states without error.
The second one is where a LOCC protocol can, in fact, distinguish
the states without error but not with unit probability. This can sometimes
be understood as unambiguous state discrimination where one of the
outcomes is inconclusive, but all others are conclusive \citep{Cosentino-2013}.
There are, however, examples that do not quite fit into this picture;
for instance, in LOCC discrimination of three Bell states, none of
the measurement outcomes is conclusive except one. 

We now explain why the IP strategy does not always work in the LOCC
scenario. Consider a LOCC set transformation problem where the input
states are mutually orthogonal. Therefore, they are either locally
indistinguishable or distinguishable. If they are locally indistinguishable,
the IP strategy clearly fails, and if they are not, the IP strategy
still may not succeed. To see this, suppose the input states are locally
distinguishable, and the output states are entangled. Since the input
states can be locally distinguished, we know there exists a LOCC protocol
that correctly identifies the input state. So in the first step, we
will be able to correctly identify the state, but in the process,
the initial state, in general, will transform into an intermediate
state, the entanglement of which will be crucial to carrying out the
next step. That is because the desired output state ought to be prepared
by deterministic LOCC from this intermediate state and a necessary
condition for this to happen is that the entanglement of the intermediate
state must be greater than or equal to that of the output state. So
if this condition is not satisfied, which may well be the case, we
can never achieve our goal. Therefore, while orthogonal input states
make the general problem trivial, the same cannot be said in the LOCC
scenario. 

\subsection*{Entanglement}

The entanglement of the states under consideration is another crucial
element that differentiates the LOCC scenario from the general problem.
Since the input and output states could be entangled, the LOCC restrictions
on entanglement transformations come into play. For example, we know
that a product state cannot be converted into an entangled state by
LOCC with nonzero probability. So if an input state is a product and
the corresponding output state is entangled, the desired LOCC conversion
at the set level obviously cannot happen. 

One of the results we will extensively use in our analysis is Nielsen's
theorem \citep{Nielsen-99} which provides a necessary and sufficient
condition for converting a bipartite pure state to another using deterministic
LOCC. 

\emph{Schmidt decomposition.} Every bipartite pure state $\mathcal{\left|\psi\right\rangle \in H}_{A}\otimes\mathcal{H}_{B}$
can be written as a linear superposition of biorthogonal product states
known as the Schmidt decomposition \citep{entanglement-review}: 
\begin{flalign}
\left|\psi\right\rangle  & =\sum_{i=1}^{r}\lambda_{i}\left|a_{i}\right\rangle \left|b_{i}\right\rangle ,\hspace{1em}\lambda_{i}\geq\lambda_{i+1}\geq0,\;\sum_{i=1}^{r}\lambda_{i}^{2}=1\label{Schmidt decomposition}
\end{flalign}
where $r\leq\min\left\{ \dim\mathcal{H}_{A},\dim\mathcal{H}_{B}\right\} $
is the Schmidt number, $\left\{ \lambda_{i}\right\} $ are the Schmidt
coefficients, and $\left\{ \left|a_{i}\right\rangle \right\} $ and
$\left\{ \left|b_{i}\right\rangle \right\} $ are orthonormal bases
of $\mathcal{H}_{A}$ and $\mathcal{H}_{B}$ respectively. 
\begin{thm}
\label{(Nielsen)} (Nielsen) Let $\left|\chi\right\rangle $ and $\left|\eta\right\rangle $
be two bipartite states with Schmidt coefficients given by $\left\{ \alpha_{i}\right\} _{i=1}^{n}$
and $\left\{ \beta_{i}\right\} _{i=1}^{m}$, where $\alpha_{i}\geq\alpha_{i+1}$
for $i=1,\dots,n-1$, and $\beta_{j}\geq\beta_{j+1}$ for $j=1,\dots,m-1$.
Then there exists a deterministic LOCC transformation $\left|\chi\right\rangle \rightarrow\left|\eta\right\rangle $
if and only if 
\begin{alignat}{1}
\sum_{i=1}^{l}\alpha_{i}^{2} & \leq\sum_{i=1}^{l}\beta_{i}^{2}\hspace{1em}\forall l=1,\dots,\min\left\{ n,m\right\} .\label{Nielsen-1}
\end{alignat}
\end{thm}
This condition can be easily checked for any pair of pure bipartite
states. So if the problem involves such states, Nielsen's theorem
needs to be satisfied for all input-output pairs of states. However,
as we will see, even if that is the case, it does not guarantee that
the desired LOCC protocol exists at the set level. 

\subsection*{Distinguishability}

How well a given set of states can be distinguished can be quantified
by distinguishability measures, such as fidelity or the average probability
of success. So let us first explain the concept of a distinguishability
measure through a specific one, namely, fidelity \citep{BN-2013,Navascues-2008,Fuchs-2004}.
This brief discussion will capture the essential elements of any well-defined
measure. 

We can associate a given set of states $S=\left\{ \rho_{1},\dots,\rho_{n}\right\} $
with a probability distribution $\bm{p}=\left\{ p_{1},\dots,p_{n}\right\} $,
where we assume that $\rho_{i}$ appears with probability $p_{i}$.
Then, for a measurement (POVM) $\mathbb{M}=\left\{ M_{a}\right\} $
and a guessing strategy $\mathbb{G}:a\rightarrow\tau_{a}$ (maps the
measurement outcomes to new quantum states), the average fidelity
is given by 

\begin{alignat}{1}
\bm{F}\left(S;\mathbb{M},\mathbb{G}\right) & =\sum_{a,i}p_{i}\text{Tr}\left(\rho_{i}M_{a}\right)\text{Tr}\left(\rho_{i}\tau_{a}\right),\label{average fidelity}
\end{alignat}
where $0\leq\bm{F}\left(S;\mathbb{M},\mathbb{G}\right)\leq1$. The
fidelity (global optimum) is now obtained by optimizing the average
fidelity over all measurements and guessing strategies:
\begin{align}
\bm{F}\left(S\right) & =\sup_{\mathbb{M},\mathbb{G}}\bm{F}\left(S;\mathbb{M},\mathbb{G}\right).\label{fidelity-1}
\end{align}
Conceptually, fidelity quantifies how much can we learn about the
system that has been prepared in a state $\rho_{i}$ with probability
$p_{i}$. Note that if the states are orthogonal we can always correctly
identify the input state and then $\bm{F}\left(S\right)=1$ for any
distribution of prior probabilities. On the other hand, if the states
are nonorthogonal, we have $\bm{F}\left(S\right)<1$. The prior probabilities,
however, are crucial in computing $\bm{F}\left(S\right)$, for if
we change the prior probabilities, $\bm{F}\left(S\right)$ will also
change, in general. 

The above definition of fidelity can be suitably modified to accommodate
composite systems and LOCC, SEP, or PPT measurements. Suppose that
$S$ now corresponds to a set of bipartite or multipartite states.
Clearly, the definition of (global) fidelity as given by Eq.$\,$(\ref{fidelity-1})
will not change. The LOCC fidelity or local fidelity (LOCC or local
optimum), on the other hand, can be obtained simply by restricting
the class of measurements to LOCC. In particular, 
\begin{alignat}{1}
\bm{F}_{l}\left(S;\mathbb{M},\mathbb{G}\right) & =\sum_{a,i}p_{i}\text{Tr}\left(\rho_{i}M_{a}\right)\text{Tr}\left(\rho_{i}\tau_{a}\right),\;\text{\ensuremath{\mathbb{M}\in\text{LOCC}}};\label{average fidelity-1}\\
\bm{F}_{l}\left(S\right) & =\sup_{\mathbb{M},\mathbb{G}}\bm{F}_{l}\left(S;\mathbb{M},\mathbb{G}\right),\;\mathbb{M}\in\text{LOCC}.
\end{alignat}
So LOCC fidelity is obtained by optimizing the average LOCC fidelity
over all LOCC measurements and guessing strategies. Note that for
a set of locally indistinguishable orthogonal states, $\bm{F}\left(S\right)=1$
but $\bm{F}_{l}\left(S\right)<1$. 

In this paper, however, we will not choose any particular measure
for our analysis as any well-defined one will serve our purpose. Let
$\bm{D}$ be a measure of distinguishability, assumed to be well-defined
for any given class of measurements. Then for a given set $S$ of
bipartite or multipartite states with associated probabilities given
by $\bm{p}=\left\{ p_{1},\dots,p_{n}\right\} $, let $\bm{D}_{g}\left(S\right)$
and $\bm{D}_{l}\left(S\right)$ denote the global and LOCC optimum,
respectively. They naturally satisfy $\bm{D}_{l}\left(S\right)\leq\bm{D}_{g}\left(S\right)$
for any choice of probability distribution $\bm{p}$. 

Suppose under some deterministic quantum operation $\bm{\Lambda}$
we have $S_{1}\rightarrow S_{2}$ where $S_{1}$ and $S_{2}$ are
two sets of bipartite or multipartite states. Since distinguishability
cannot increase under deterministic quantum operations we have the
following observations: 
\begin{itemize}
\item If $\bm{\Lambda}$ belongs to the set of all quantum operations, then
for any $\bm{p}$ it holds that 
\begin{flalign}
\bm{D}_{g}\left(S_{1}\right) & \geq\bm{D}_{g}\left(S_{2}\right).\label{eglobal D inequality}
\end{flalign}
\item If $\bm{\Lambda}\in\text{LOCC}$ then we must also have
\begin{equation}
\bm{D}_{l}\left(S_{1}\right)\geq\bm{D}_{l}\left(S_{2}\right)\label{local D inequality}
\end{equation}
for any $\bm{p}$. 
\end{itemize}
The second inequality becomes significant when both the input and
output sets consist of mutually orthogonal states, for the first inequality
cannot help in this case. In such cases, one may apply the second
equality to rule out some kinds of transformations. For example, if
$S_{1}$ consists of locally indistinguishable states, but the states
in $S_{2}$ are locally distinguishable, then it is not possible to
achieve $S_{1}\rightarrow S_{2}$. More generally, if both are locally
indistinguishable and $\bm{D}_{l}\left(S_{1}\right)<\bm{D}_{l}\left(S_{2}\right)$,
then $S_{1}\nrightarrow S_{2}$ under LOCC. However, the second inequality
will not give us any useful information if both sets are locally distinguishable. 

\section{LOCC transformations between sets of pure states}

The discussions in the previous section show that the LOCC problem
is indeed special in some ways, for the states could be entanglement,
orthogonal, and locally distinguishable or indistinguishable. We will
now examine the question of the existence of a deterministic LOCC
transformation between two sets of bipartite pure states. 

\subsection*{Necessary conditions}

Necessary conditions are obtained from LOCC constraints on state transformation,
entanglement, and distinguishability. 
\begin{prop}
\label{necessary-conditions} Consider a bipartite quantum system
$\mathcal{H}=\mathcal{H}_{A}\otimes\mathcal{H}_{B}$ of finite dimensions.
Let $S_{\psi}=\left\{ \left|\psi_{1}\right\rangle ,\left|\psi_{2}\right\rangle ,\dots,\left|\psi_{n}\right\rangle \right\} $
and $S_{\phi}=\left\{ \left|\phi_{1}\right\rangle ,\left|\phi_{2}\right\rangle ,\dots,\left|\phi_{n}\right\rangle \right\} $
be the sets of input and output states, respectively. Suppose there
exists a deterministic LOCC under which $S_{\psi}\rightarrow S_{\phi}$.
Then: 

\noindent (a) Theorem \ref{(Nielsen)} holds for every pair $\left(\left|\psi_{i}\right\rangle ,\left|\phi_{i}\right\rangle \right)$,
$i=1,\dots,n$; 

\noindent (b) For any probability distribution $\bm{p}=\left\{ p_{1},p_{2},\dots,p_{n}\right\} $,
$0<p_{i}<1$, $i=1,\dots,n$, and the density operators
\begin{flalign}
\rho_{\psi,\bm{p}} & =\sum_{i=1}^{n}p_{i}\left|\psi_{i}\right\rangle \left\langle \psi_{i}\right|,\label{rho-psi-p}\\
\rho_{\phi,\bm{p}} & =\sum_{i=1}^{n}p_{i}\left|\phi_{i}\right\rangle \left\langle \phi_{i}\right|,\label{rho-phi-p}
\end{flalign}
it holds that 
\begin{alignat}{1}
\bm{E}\left(\rho_{\psi,\bm{p}}\right) & \geq\bm{E}\left(\rho_{\phi,\bm{p}}\right),\label{entanglement-inequality}
\end{alignat}
where $\bm{E}$ is any well-defined entanglement measure; 

\noindent (c) For any probability distribution $\bm{p}=\left\{ p_{1},p_{2},\dots,p_{n}\right\} $,
$0<p_{i}<1$, $i=1,\dots,n$, and the ensembles $S_{\psi,\bm{p}}=\left\{ p_{i},\left|\psi_{i}\right\rangle \right\} $
and $S_{\phi,\bm{p}}=\left\{ p_{i},\left|\phi_{i}\right\rangle \right\} $,
it holds that 
\begin{alignat}{1}
\bm{D}\left(S_{\psi,\bm{p}}\right) & \geq\bm{D}\left(S_{\phi,\bm{p}}\right),\label{D-inequality}
\end{alignat}
where $\bm{D}$ is a well-defined measure of distinguishability.
\end{prop}
\begin{proof}
There is nothing in particular to discuss about $\left(a\right)$
which is obvious. Note however that $\left(a\right)$ is stronger
than a plausible alternative: $\bm{E}\left(\psi_{i}\right)\geq\bm{E}\left(\phi_{i}\right)$,
$i=1,\dots,n$. That is because if $\left(a\right)$ is satisfied
then $\bm{E}\left(\psi_{i}\right)\geq\bm{E}\left(\phi_{i}\right)$
for every $i$, but the converse does not hold. In particular, there
exists states $\left|\psi\right\rangle $ and $\left|\phi\right\rangle $
satisfying $\bm{E}\left(\psi\right)\geq\bm{E}\left(\phi\right)$ but
$\left|\psi\right\rangle \nrightarrow\left|\phi\right\rangle $ under
deterministic LOCC \citep{Nielsen-99}. 

We now come to $\left(b\right)$. Since $S_{\text{\ensuremath{\psi}}}\rightarrow S_{\phi}$
is possible by a deterministic LOCC, then for any probability distribution
$\bm{p}=\left\{ p_{1},p_{2},\dots,p_{n}\right\} $, \emph{$0<p_{i}<1$,
$i=1,\dots,n$,} it holds that 
\begin{flalign*}
\left\{ \left(p_{1},\left|\psi_{1}\right\rangle \right),\left(p_{2},\left|\psi_{2}\right\rangle \right),\dots,\left(p_{n},\left|\psi_{n}\right\rangle \right)\right\}  & \rightarrow\left\{ \left(p_{1},\left|\phi_{1}\right\rangle \right),\left(p_{2},\left|\phi_{2}\right\rangle \right),\dots,\left(p_{n},\left|\phi_{n}\right\rangle \right)\right\} ,
\end{flalign*}
or equivalently, 
\begin{flalign*}
\rho_{\psi,\bm{p}} & \rightarrow\rho_{\phi,\bm{p}},
\end{flalign*}
where $\rho_{\psi,\bm{p}}$ and $\rho_{\phi,\bm{p}}$ are given by
(\ref{rho-psi-p}) and (\ref{rho-phi-p}) respectively. Because the
entanglement of a state cannot increase under deterministic LOCC,
(\ref{entanglement-inequality}) must be satisfied for all $\bm{p}$
and any well-defined measure of entanglement $\bm{E}$. 

To prove $\left(c\right)$, once again first note that $S_{\text{\ensuremath{\psi}}}\rightarrow S_{\phi}$
by a deterministic LOCC. The proof then follows from the fact that
the distinguishability of given set of states cannot increase under
a deterministic quantum operation, LOCC or otherwise. In other words,
the ensemble $S_{\psi,\bm{p}}$ cannot be less distinguishable than
$S_{\phi,\bm{p}}$ for any given probability distribution $\bm{p}$.
Hence, (\ref{D-inequality}) must be satisfied for any well-defined
measure of distinguishability $\bm{D}$ and all $\bm{p}$. 
\end{proof}
Let us now look at the implications of $\left(a\right)$, $\left(b\right)$,
and $\left(c\right)$. Specifically, we would like to know whether
they are independent of each other and what possible relations exist
between them. These questions are better answered by considering sets
of input-output pairs that satisfy each of them. 
\begin{prop}
\label{set-relations-A} Let $S_{a}$, $S_{b}$, and $S_{c}$ denote
the sets of all input-output pairs satisfying conditions $\left(a\right)$,
$\left(b\right)$, and $\left(c\right)$, respectively. The relations
\begin{flalign}
S_{x}\setminus\left(S_{y}\cup S_{z}\right) & \neq\emptyset\label{x-yz}
\end{flalign}

hold for $x\neq y\neq z\in\left\{ a,b,c\right\} $. 
\end{prop}
\begin{proof}
\begin{flushleft}
Equation (\ref{x-yz}) tells us that for any choice of $x\neq y\neq z\in\left\{ a,b,c\right\} $,
there exist pairs of input and output sets $S_{\psi}$ and $S_{\phi}$
that satisfy $\left(x\right)$ but do not satisfy $\left(y\right)$
or $\left(z\right)$. We will now prove (\ref{x-yz}) for each case. 
\par\end{flushleft}
\begin{flushleft}
\[
\bm{S_{a}\setminus\left(S_{b}\cup S_{c}\right)\neq\emptyset}
\]
\par\end{flushleft}
$\hspace{1em}\hspace{1em}\hspace{1em}$%
\begin{tabular}{|c||c||c||c|}
\hline 
 $S_{\psi}$  &  $S_{\phi}$ & \multirow{1}{*}{Satisfied} & Not satisfied \tabularnewline[\doublerulesep]
\hline 
$\left|\psi_{1}\right\rangle =\frac{1}{\sqrt{2}}\left(\left|00\right\rangle +\left|11\right\rangle \right)$ & $\left|\phi_{1}\right\rangle =\frac{1}{\sqrt{2}}\left(\left|01\right\rangle +\left|10\right\rangle \right)$ &  & \tabularnewline
$\left|\psi_{2}\right\rangle =\frac{1}{\sqrt{2}}\left(\left|00\right\rangle -\left|11\right\rangle \right)$ & $\left|\phi_{2}\right\rangle =\left|00\right\rangle $ & $\left(a\right)$ & $\left(b\right)$ and $\left(c\right)$\tabularnewline
$\left|\psi_{3}\right\rangle =\left|01\right\rangle $ & $\left|\phi_{3}\right\rangle =\left|11\right\rangle $ &  & \tabularnewline[\doublerulesep]
\hline 
\end{tabular}\\
\\
\\
$\left(a\right)$ is satisfied: Straightforward to check. \\
$\left(b\right)$ is not satisfied: Consider the probability distribution
of the form $\bm{p}=\left\{ p,p,1-2p\right\} $, $0<p<1/2$. Then
one finds that $\rho_{\psi,\bm{p}}$ is separable for all $0<p<1/2$,
but $\rho_{\phi,\bm{p}}$ is entangled for $4/9<p<1/2$. The latter
can be obtained by applying the partial transposition criterion \citep{Partial=00003Dtranspose}
or by computing the concurrence \citep{Concurrence}. Thus the inequality
(\ref{entanglement-inequality}) does not hold whenever $4/9<p<1/2$.
Hence, $\left(b\right)$ is not satisfied. \\
$\left(c\right)$ is not satisfied: To prove this we will use the
following result proved in \citep{Walgate-Hardy-2002}: 
\begin{lem}
Three orthogonal two-qubit pure states can be distinguished by LOCC
if and only if at least two of those states are product states. 
\end{lem}
By applying the above lemma it immediately follows that $S_{\psi}$
is locally indistinguishable but $S_{\phi}$ is locally distinguishable.
So inequality (\ref{D-inequality}) is violated for any LOCC measure
of distinguishability (discussions on such measures can be found in
\citep{BN-2013}). Hence, $\left(c\right)$ is not satisfied.

\pagebreak

\[
\bm{S_{b}\setminus\left(S_{a}\cup S_{c}\right)\neq\emptyset}
\]
\begin{flushleft}
\begin{tabular}{|c||c||c||c|}
\hline 
 $S_{\psi}$  &  $S_{\phi}$ & Satisfied & Not satisfied \tabularnewline
\hline 
\noalign{\vskip\doublerulesep}
$\left|\psi_{1}\right\rangle =\frac{1}{\sqrt{2}}\left(\left|00\right\rangle +\left|11\right\rangle \right)$ & $\left|\phi_{1}\right\rangle =\sqrt{\frac{4}{5}}\left|00\right\rangle +\sqrt{\frac{1}{10}}\left|11\right\rangle +\sqrt{\frac{1}{10}}\left|22\right\rangle $ &  & \tabularnewline
$\left|\psi_{2}\right\rangle =\frac{1}{\sqrt{2}}\left(\left|22\right\rangle +\left|33\right\rangle \right)$ & $\left|\phi_{2}\right\rangle =\left|13\right\rangle $ & $\left(b\right)$ & $\left(a\right)$ and $\left(c\right)$\tabularnewline
$\left|\psi_{3}\right\rangle =\frac{1}{\sqrt{2}}\left(\left|22\right\rangle -\left|33\right\rangle \right)$ & $\left|\phi_{3}\right\rangle =\left|23\right\rangle $ &  & \tabularnewline
$\left|\psi_{4}\right\rangle =\frac{1}{\sqrt{2}}\left(\left|23\right\rangle +\left|32\right\rangle \right)$ & $\left|\phi_{4}\right\rangle =\left|33\right\rangle $ &  & \tabularnewline[\doublerulesep]
\hline 
\end{tabular}
\par\end{flushleft}
\smallskip
\begin{flushleft}
$\left(b\right)$ is satisfied: Consider the density operators $\rho_{\psi,\bm{p}}$
and $\rho_{\phi,\bm{p}}$ with $\bm{p}=\left\{ p_{1},p_{2},p_{3},p_{4}\right\} $,
$0<p_{i}<1$, $i=1,\dots,4$. Let $\bm{E}$ denote the entanglement
of formation \citep{entanglement-review} measured in ebits. We will
show that $\bm{E}\left(\rho_{\psi,\bm{p}}\right)\geq p_{1}$ but $\bm{E}\left(\rho_{\phi,\bm{p}}\right)<p_{1}$. 
\par\end{flushleft}
First, note that we can distill at least $p_{1}$ ebit from $\rho_{\psi,\bm{p}}$.
To see this, consider the following LOCC protocol. Alice performs
a binary measurement that distinguishes between the two orthogonal
subspaces spanned by $\left\{ \left|0\right\rangle ,\left|1\right\rangle \right\} $
and $\left\{ \left|2\right\rangle ,\left|3\right\rangle \right\} $.
If she obtains the first outcome, she and Bob end up sharing the state
$\left|\psi_{1}\right\rangle $ which has one ebit of entanglement.
The probability of getting this outcome is $p_{1}$, so the distillable
entanglement is at least $p_{1}$ ebit. Since distillable entanglement
is a lower bound on the entanglement of formation \citep{entanglement-review},
we have $\bm{E}\left(\rho_{\psi,\bm{p}}\right)\geq p_{1}$. Now,
\begin{alignat*}{1}
\bm{E}\left(\rho_{\phi,\bm{p}}\right) & \leq\sum_{i=1}^{4}p_{i}\bm{E}\left(\phi_{i}\right),\\
 & =p_{1}\bm{E}\left(\phi_{1}\right),\\
 & <p_{1}
\end{alignat*}
since $\bm{E}\left(\phi_{1}\right)<1$ (one could easily check) and
$\bm{E}\left(\phi_{i}\right)=0$ for $i=2,3,4$. So the inequality
(\ref{entanglement-inequality}) is satisfied for all $\bm{p}$. \\
$\left(a\right)$ is not satisfied: This follows from the observation
that Nielsen criterion is violated for the pair $\left(\left|\psi_{1}\right\rangle ,\left|\phi_{1}\right\rangle \right)$.
\\
$\left(c\right)$ is not satisfied: $S_{\psi}$ is locally indistinguishable,
for it contains a locally indistinguishable subset $\left\{ \left|\psi_{2}\right\rangle ,\left|\psi_{3}\right\rangle ,\left|\psi_{4}\right\rangle \right\} $
\citep{B-IQC-2015}, whereas $S_{\phi}$ is locally distinguishable,
for a local measurement in the computational basis perfectly distinguishes
the states.

\pagebreak

\[
\bm{S_{c}\setminus\left(S_{a}\cup S_{b}\right)\neq\emptyset}
\]

$\hspace{1em}\hspace{1em}\hspace{1em}\hspace{1em}\hspace{1em}$%
\begin{tabular}{|c||c||c||c|}
\hline 
\noalign{\vskip\doublerulesep}
$S_{\psi}$  & $S_{\phi}$ & \multirow{1}{*}{Satisfied} & Not satisfied \tabularnewline[\doublerulesep]
\hline 
$\left|\psi_{1}\right\rangle =\left|01\right\rangle $ & $\left|\phi_{1}\right\rangle =\frac{1}{\sqrt{2}}\left(\left|00\right\rangle +\left|11\right\rangle \right)$ & $\left(c\right)$ & $\left(a\right)$ and $\left(b\right)$\tabularnewline
$\left|\psi_{2}\right\rangle =\left|10\right\rangle $ & $\left|\phi_{2}\right\rangle =\left|00\right\rangle $ &  & \tabularnewline
\hline 
\end{tabular}\\
\\
\\
The proof that $\left(c\right)$ is satisfied is simple: $S_{\psi}$
is a distinguishable set, whereas $S_{\phi}$ contains nonorthogonal
states, which cannot be perfectly distinguished. \\
It is also easy to see that $\left(a\right)$ and $\left(b\right)$
are not satisfied. The first follows from the observation that $\left|\psi_{1}\right\rangle $
is a product state but $\left|\phi_{1}\right\rangle $ is entangled.
And the second follows from the following easily checkable properties:
$\rho_{\psi,\bm{p}}$ is separable but $\rho_{\phi,\bm{p}}$ is entangled
(apply the partial-transposition criterion \citep{entanglement-review}). 
\end{proof}
We will now show that for any choice of $x\neq y\neq z\in\left\{ a,b,c\right\} $,
one can find pairs of input and output sets that satisfy both $\left(x\right)$
and $\left(y\right)$ but not $\left(z\right)$. 
\begin{prop}
\label{set-relations-B} Let $S_{a}$, $S_{b}$, and $S_{c}$ denote
the sets of all input-output pairs satisfying conditions $\left(a\right)$,
$\left(b\right)$, and $\left(c\right)$, respectively. The relations
\begin{flalign}
\left(S_{x}\cap S_{y}\right)\setminus S_{z} & \neq\emptyset\label{xy-z}
\end{flalign}

hold for $x\neq y\neq z\in\left\{ a,b,c\right\} $. 
\end{prop}
\begin{proof}
As before, we will prove (\ref{xy-z}) for each case. 
\[
\bm{\left(S_{a}\cap S_{b}\right)\setminus S_{c}\neq\emptyset}
\]

$\hspace{1em}\hspace{1em}\hspace{1em}\hspace{1em}\hspace{1em}$%
\begin{tabular}{|c||c||c||c|}
\hline 
\noalign{\vskip\doublerulesep}
$S_{\psi}$  & $S_{\phi}$ & \multirow{1}{*}{Satisfied} & Not satisfied \tabularnewline[\doublerulesep]
\hline 
$\left|\psi_{1}\right\rangle =\frac{1}{\sqrt{2}}\left(\left|00\right\rangle +\left|11\right\rangle \right)$ & $\left|\phi_{1}\right\rangle =\left|01\right\rangle $ & $\left(a\right)$ and $\left(b\right)$  & $\left(c\right)$ \tabularnewline
$\left|\psi_{2}\right\rangle =\left|00\right\rangle $ & $\left|\phi_{2}\right\rangle =\left|10\right\rangle $ &  & \tabularnewline[\doublerulesep]
\hline 
\end{tabular}\\
\\
\\
Clearly, $\left(a\right)$ is satisfied; $\left(b\right)$ is also
satisfied because $\rho_{\psi,\bm{p}}$ is entangled but $\rho_{\phi,\bm{p}}$
is separable for all $\bm{p}=\left\{ p,1-p\right\} $, $p\in\left(0,1\right)$. 

Note, however, that $\left(c\right)$ cannot be satisfied. That is
because $S_{\psi}$ contains nonorthogonal states, which cannot be
perfectly distinguished, whereas the states in $S_{\phi}$ are mutually
orthogonal and can be perfectly distinguished.

\[
\bm{\left(S_{a}\cap S_{c}\right)\setminus S_{b}\neq\emptyset}
\]

$\hspace{1em}\hspace{1em}$%
\begin{tabular}{|c||c||c||c|}
\hline 
 $S_{\psi}$  &  $S_{\phi}$ & \multirow{1}{*}{Satisfied} & Not satisfied \tabularnewline[\doublerulesep]
\hline 
$\left|\psi_{1}\right\rangle =\frac{1}{\sqrt{2}}\left(\left|00\right\rangle +\left|11\right\rangle \right)$ & $\left|\phi_{1}\right\rangle =\alpha\left|00\right\rangle +\beta\left|11\right\rangle $ & $\left(a\right)$ and $\left(c\right)$ & $\left(b\right)$\tabularnewline
$\left|\psi_{2}\right\rangle =\frac{1}{\sqrt{2}}\left(\left|00\right\rangle -\left|11\right\rangle \right)$ & $\left|\phi_{2}\right\rangle =\beta\left|00\right\rangle +\alpha\left|11\right\rangle $ &  & \tabularnewline[\doublerulesep]
\hline 
\end{tabular}\\
\\
where $\alpha>\beta>0$ and $\alpha^{2}+\beta^{2}=1$. \\
\\
Nielsen's criterion is satisfied for each $i=1,2$, so $\left(a\right)$
is satisfied. Condition $\left(c\right)$ is also satisfied because
the input states are distinguishable but the output states, being
nonorthogonal, are not. 

On the other hand, for $p=\frac{1}{2}$, one finds that $\rho_{\psi,\bm{p}}=p\left|\psi_{1}\right\rangle \left\langle \psi_{1}\right|+\left(1-p\right)\left|\psi_{2}\right\rangle \left\langle \psi_{2}\right|$
is separable but $\rho_{\phi,\bm{p}}=p\left|\phi_{1}\right\rangle \left\langle \phi_{1}\right|+\left(1-p\right)\left|\phi_{2}\right\rangle \left\langle \phi_{2}\right|$
is entangled. So the inequality (\ref{entanglement-inequality}) cannot
be satisfied for all probability distributions. Hence, $\left(b\right)$
is not satisfied. 

\[
\bm{\left(S_{b}\cap S_{c}\right)\setminus S_{a}\neq\emptyset}
\]
\begin{flushleft}
\begin{tabular}{|c||c||c||c|}
\hline 
 $S_{\psi}$  &  $S_{\phi}$ & Satisfied & Not satisfied \tabularnewline
\hline 
\noalign{\vskip\doublerulesep}
$\left|\psi_{1}\right\rangle =\frac{1}{\sqrt{2}}\left(\left|00\right\rangle +\left|11\right\rangle \right)$ & $\left|\phi_{1}\right\rangle =\sqrt{\frac{4}{5}}\left|00\right\rangle +\sqrt{\frac{1}{10}}\left|11\right\rangle +\sqrt{\frac{1}{10}}\left|22\right\rangle $ & $\left(b\right)$ and $\left(c\right)$ & $\left(a\right)$\tabularnewline
$\left|\psi_{2}\right\rangle =\frac{1}{\sqrt{2}}\left(\left|22\right\rangle +\left|33\right\rangle \right)$ & $\left|\phi_{2}\right\rangle =\left|33\right\rangle $ &  & \tabularnewline
\hline 
\end{tabular}
\par\end{flushleft}
\begin{flushleft}
\smallskip
\par\end{flushleft}
\begin{flushleft}
Let $\bm{E}$ denote the entanglement of formation measured in ebits.
Then: $\bm{E}\left(\psi_{1}\right)=\bm{E}\left(\psi_{2}\right)=1$,
$\bm{E}\left(\phi_{1}\right)<1$, and $\bm{E}\left(\phi_{2}\right)=0$.
\\
$\left(b\right)$ is satisfied: Consider the density matrices $\rho_{\psi,\bm{p}}$
and $\rho_{\phi,\bm{p}}$, where $\bm{p}=\left\{ p,1-p\right\} $,
$p\in\left(0,1\right)$. First, note that that one ebit can be distilled
from $\rho_{\psi,\bm{p}}$ for any $\bm{p}$ using deterministic LOCC.
The protocol is similar to the one discussed earlier, so we omit the
details. Because distillable entanglement is a lower bound on the
entanglement of formation, we have $\bm{E}\left(\rho_{\psi,\bm{p}}\right)\geq1$.
On the other hand, $\bm{E}\left(\rho_{\phi,\bm{p}}\right)\leq p\bm{E}\left(\phi_{1}\right)+\left(1-p\right)\bm{E}\left(\phi_{2}\right)<p<1$.
Therefore, $\bm{E}\left(\rho_{\psi,\bm{p}}\right)>\bm{E}\left(\rho_{\phi,\bm{p}}\right)$
for all $\bm{p}$. So $\left(b\right)$ is satisfied. \\
$\left(c\right)$ is satisfied: Both sets are distinguishable, locally
or globally.\\
$\left(a\right)$ is not satisfied: That is because the transformation
$\left|\psi_{1}\right\rangle \rightarrow\left|\phi_{1}\right\rangle $
is not possible by deterministic LOCC as Nielsen's criterion is violated. 
\par\end{flushleft}
\end{proof}
\begin{rem*}
Propositions \ref{set-relations-A} and \ref{set-relations-B} establish
independence of the conditions $\left(a\right)$, $\left(b\right)$,
and $\left(c\right)$. 
\end{rem*}
\begin{prop}
\label{set-relation-C} Let $S_{a}$, $S_{b}$, and $S_{c}$ denote
the sets of all input-output pairs satisfying conditions $\left(a\right)$,
$\left(b\right)$, and $\left(c\right)$, respectively. Then the following
relation holds:
\begin{flalign}
\mathcal{S}_{a}\cap\mathcal{S}_{b}\cap\mathcal{S}_{c} & \neq\emptyset.\label{xyz}
\end{flalign}
\end{prop}
\begin{proof}
There are many examples of input-output pairs $\left(S_{\psi},S_{\phi}\right)$
for which $S_{\psi}\rightarrow S_{\phi}$ by deterministic LOCC. Any
such pair satisfies all three conditions $\left(a\right)$, $\left(b\right)$,
and $\left(c\right)$ and therefore belongs to the set $\mathcal{S}_{a}\cap\mathcal{S}_{b}\cap\mathcal{S}_{c}$.
A simple example is where the output states $\left\{ \left|\phi_{i}\right\rangle \right\} $
are related to the corresponding input states $\left|\psi_{i}\right\rangle $
by some fixed local unitary operator: 
\begin{flalign*}
\left|\phi_{i}\right\rangle  & =U\otimes V\left|\psi_{i}\right\rangle ,i=1,\dots,n.
\end{flalign*}
A somewhat more nontrivial and instructive example is the following
pair of $S_{\psi}$ and $S_{\phi}$:
\begin{eqnarray*}
\left|\psi_{1}\right\rangle =\frac{1}{\sqrt{2}}\left(\left|00\right\rangle +\left|11\right\rangle \right) & \hspace{1em}\hspace{1em} & \left|\phi_{1}\right\rangle =\sqrt{\frac{4}{5}}\left|00\right\rangle +\sqrt{\frac{1}{5}}\left|11\right\rangle \\
\left|\psi_{2}\right\rangle =\frac{1}{\sqrt{2}}\left(\left|22\right\rangle +\left|33\right\rangle \right) & \hspace{1em}\hspace{1em} & \left|\phi_{2}\right\rangle =\left|22\right\rangle 
\end{eqnarray*}
It is enough to show there exists a deterministic LOCC transformation
under which $S_{\psi}\rightarrow S_{\phi}$. The protocol is an IP
strategy. First, Alice (or Bob) performs an orthogonal measurement
that distinguishes the subspaces spanned by $\left\{ \left|0\right\rangle ,\left|1\right\rangle \right\} $
and $\left\{ \left|2\right\rangle ,\left|3\right\rangle \right\} $.
The outcome corresponding to $\left\{ \left|0\right\rangle ,\left|1\right\rangle \right\} $
implies that they hold $\left|\psi_{1}\right\rangle $, otherwise,
$\left|\psi_{2}\right\rangle $. Note that the measurement does not
change the input state in any way. Once they correctly identify the
input, they transform it to the desired output state using a deterministic
LOCC protocol. That such a protocol exists for each input-output pair
follows from Theorem \ref{(Nielsen)}.
\end{proof}
\begin{rem*}
Note that, although the output set contains an entangled state, the
IP strategy works in this case because the local protocol identifying
the input state is \emph{nondestructive}; that is, the input state
remains intact during the discrimination process.
\end{rem*}

\section{Are the necessary conditions sufficient? }

It is reasonable to ask whether the conditions in Proposition \ref{necessary-conditions}
are also sufficient. The answer, unfortunately, turns out to be no. 
\begin{prop}
\label{insufficient} The set of necessary conditions in Proposition
\ref{necessary-conditions} is not sufficient. 
\end{prop}
\begin{proof}
Consider the input and output sets $S_{\psi}=\left\{ \left|\psi_{1}\right\rangle ,\left|\psi_{2}\right\rangle \right\} $
and $S_{\phi}=\left\{ \left|\phi_{1}\right\rangle ,\left|\phi_{2}\right\rangle \right\} $,
where 
\begin{eqnarray*}
\left|\psi_{1}\right\rangle =\frac{1}{\sqrt{2}}\left(\left|00\right\rangle +\left|11\right\rangle \right) & \hspace{1em}\hspace{1em} & \left|\phi_{1}\right\rangle =\alpha\left|00\right\rangle +\beta\left|11\right\rangle \\
\left|\psi_{2}\right\rangle =\frac{1}{\sqrt{2}}\left(\left|00\right\rangle -\left|11\right\rangle \right) & \hspace{1em}\hspace{1em} & \left|\phi_{2}\right\rangle =\beta\left|00\right\rangle -\alpha\left|11\right\rangle 
\end{eqnarray*}
where $\alpha>\beta>0$ and $\alpha^{2}+\beta^{2}=1$. 

First, we will show all three conditions in Proposition \ref{necessary-conditions}
are satisfied. Then we will prove that the desired set transformation
is not possible by deterministic LOCC. 

Condition $\left(a\right)$ holds as Nielsen's theorem is satisfied
in each case. To show that $\left(b\right)$ holds we proceed as follows.
Consider the density operators $\rho_{\psi,\bm{p}}$ and $\rho_{\phi,\bm{p}}$,
where $\bm{p}=\left\{ p,1-p\right\} $, $p\in\left(0,1\right)$. The
respective concurrences are given by 
\begin{alignat}{1}
\bm{C}\left(\rho_{\psi,\bm{p}}\right) & =\begin{array}{c}
\left|1-2p\right|,\;p\neq\frac{1}{2}\\
0,\;p=\frac{1}{2}
\end{array}\label{concurrence-rho-psi}\\
\nonumber \\
\bm{C}\left(\rho_{\phi,\bm{p}}\right) & =\begin{array}{c}
2\alpha\beta\left|1-2p\right|,\;p\neq\frac{1}{2}\\
0,\;p=\frac{1}{2}
\end{array}\label{concurrence-rho-phi}
\end{alignat}
It follows from (\ref{concurrence-rho-psi}) and (\ref{concurrence-rho-phi})
that $\bm{C}\left(\rho_{\psi,\bm{p}}\right)\geq\bm{C}\left(\rho_{\phi,\bm{p}}\right)$
for all $\bm{p}$. Therefore, $\left(b\right)$ is satisfied. Since
both sets are distinguishable, locally \citep{Walgate+2000} or globally,
condition $\left(c\right)$ is satisfied as well.

We now prove that there does not exist a separable operation $\mathcal{\bm{S}}$
satisfying $\mathcal{\bm{S}}\left(\psi_{i}\right)=\phi_{i}$, where
$\psi_{i}=\left|\psi_{i}\right\rangle \left\langle \psi_{i}\right|$
and $\phi_{i}=\left|\phi_{i}\right\rangle \left\langle \phi_{i}\right|$. 

Suppose, on the contrary, there exists a separable operation $\mathcal{\bm{S}}=\left\{ \mathcal{A}_{k}\otimes\mathcal{B}_{k}\right\} $,
where $\mathcal{A}_{k}\otimes\mathcal{B}_{k}$ are Kraus operators
satisfying $\sum_{k}\mathcal{A}_{k}^{\dagger}\mathcal{A}_{k}\otimes\mathcal{B}_{k}^{\dagger}\mathcal{B}_{k}=\mathcal{I}$,
$\mathcal{I}$ being the identity operator, such that $\mathcal{\bm{S}}\left(\psi_{i}\right)=\phi_{i}$
for $i=1,2$; that is: 
\begin{alignat}{1}
\mathcal{\bm{S}}\left(\psi_{i}\right) & =\sum_{k}\left(\mathcal{A}_{k}\otimes\mathcal{B}_{k}\right)\psi_{i}\left(\mathcal{A}_{k}^{\dagger}\otimes\mathcal{B}_{k}^{\dagger}\right)=\phi_{i},i=1,2.\label{Sep-on-inputs}
\end{alignat}
From (\ref{Sep-on-inputs}), it follows that 
\begin{alignat}{1}
\mathcal{A}_{k}\otimes\mathcal{B}_{k}\left|\psi_{1}\right\rangle  & =\mu_{k1}\left|\phi_{1}\right\rangle ,\label{muk1}\\
\mathcal{A}_{k}\otimes\mathcal{B}_{k}\left|\psi_{2}\right\rangle  & =\mu_{k2}\left|\phi_{2}\right\rangle ,\label{muk2}
\end{alignat}
where $\mu_{k1},\mu_{k2}$ are complex numbers. For a given $k$,
we need to consider two possibilities: $\mu_{k1},\mu_{k2}$ are both
nonzero or one of them is zero. 

First suppose that both $\mu_{k1}$ and $\mu_{k2}$ are nonzero. Adding
(\ref{muk1}) and (\ref{muk2}) we get 
\begin{alignat}{1}
\mathcal{A}_{k}\otimes\mathcal{B}_{k}\left(\left|\psi_{1}\right\rangle +\left|\psi_{2}\right\rangle \right) & =\mu_{k1}\left|\phi_{1}\right\rangle +\mu_{k2}\left|\phi_{2}\right\rangle .\label{muk1+muk2}
\end{alignat}
Simplifying the above equation we arrive at 
\begin{alignat}{1}
\mathcal{A}_{k}\otimes\mathcal{B}_{k}\left|00\right\rangle  & =\frac{1}{\sqrt{2}}\left[\left(\mu_{k1}\alpha+\mu_{k2}\beta\right)\left|00\right\rangle +\left(\mu_{k1}\beta-\mu_{k2}\alpha\right)\left|11\right\rangle \right].\label{Simplified muk1+muk2}
\end{alignat}
Now the LHS of (\ref{Simplified muk1+muk2}) is a product state. That
means the state on the RHS must also be a product state. Since this
state already enjoys biorthogonal decomposition, the eigenvalues of
the reduced density matrices are easy to obtain. Requiring one of
eigenvalues to be zero as the Schmidt rank of a product state is one,
we find that either of the conditions
\begin{eqnarray}
\ensuremath{\frac{\mu_{k1}}{\mu_{k2}}} & = & -\frac{\beta}{\alpha}\text{ or }\ensuremath{\frac{\alpha}{\beta}}.\label{muk1/muk2-first}
\end{eqnarray}
need to hold. 

Now subtracting (\ref{muk2}) from (\ref{muk1}) and simplifying we
find that 
\begin{alignat}{1}
\mathcal{A}_{k}\otimes\mathcal{B}_{k}\left|11\right\rangle  & =\frac{1}{\sqrt{2}}\left[\left(\mu_{k1}\alpha-\mu_{k2}\beta\right)\left|00\right\rangle +\left(\mu_{k1}\beta+\mu_{k2}\alpha\right)\left|11\right\rangle \right].\label{muk1-muk2}
\end{alignat}
Since the LHS of (\ref{muk1-muk2}) is a product state, the RHS of
(\ref{muk1-muk2}) must also be a product state, which requires us
to satisfy 
\begin{eqnarray}
\ensuremath{\frac{\mu_{k1}}{\mu_{k2}}} & = & \frac{\beta}{\alpha}\text{ or }\ensuremath{-\frac{\alpha}{\beta}.}\label{muk1/muk2-second}
\end{eqnarray}
Now observe that the conditions (\ref{muk1/muk2-first}) and (\ref{muk1/muk2-second})
both cannot be simultaneously satisfied unless $\alpha=\beta$. Since
$\alpha>\beta$ it follows that (\ref{muk1}) and (\ref{muk2}) both
cannot hold for nonzero $\mu_{k1}$ and $\mu_{k2}$. Thus the desired
set transformation is not possible by a separable operation satisfying
both (\ref{muk1}) and (\ref{muk2}) for nonzero $\mu_{k1}$ and $\mu_{k2}$. 

Next, suppose that $\mu_{k1}\neq0$ but $\mu_{k2}=0$. Then 
\begin{alignat}{1}
\mathcal{A}_{k}\otimes\mathcal{B}_{k}\left|\psi_{1}\right\rangle  & =\frac{1}{\sqrt{2}}\left(\mathcal{A}_{k}\otimes\mathcal{B}_{k}\left|00\right\rangle +\mathcal{A}_{k}\otimes\mathcal{B}_{k}\left|11\right\rangle \right)=\mu_{k1}\left|\phi_{1}\right\rangle ,\label{muk1nonzero}\\
\mathcal{A}_{k}\otimes\mathcal{B}_{k}\left|\psi_{2}\right\rangle  & =\frac{1}{\sqrt{2}}\left(\mathcal{A}_{k}\otimes\mathcal{B}_{k}\left|00\right\rangle -\mathcal{A}_{k}\otimes\mathcal{B}_{k}\left|11\right\rangle \right)=0.\label{muk2zero}
\end{alignat}
Adding (\ref{muk1nonzero}) and (\ref{muk2zero}) and simplifying
we find that the LHS is a product state, whereas the RHS is an entangled
state as $\mu_{k1}\neq0$. Hence, (\ref{muk1nonzero}) and (\ref{muk2zero})
cannot be both satisfied. A similar argument holds for the case where
$\mu_{k1}=0$ but $\mu_{k2}\neq0$. Therefore, there does not exist
a separable operation which could achieve the desired transformation. 

So we have proved that there does not exist a separable operation
satisfying (\ref{muk1}) and (\ref{muk2}), which implies that the
desired set transformation using a separable operation is impossible.
Noting that LOCC is a strict subset of separable operations, the proof
is therefore complete. 
\end{proof}
Propositions \ref{set-relations-A}\textendash \ref{insufficient}
suggest that one can now classify possible input-output pairs of sets
based on the satisfiability of the necessary conditions. This is shown
in Fig.$\,$1. 
\begin{figure}
\begin{centering}
\includegraphics[scale=0.5]{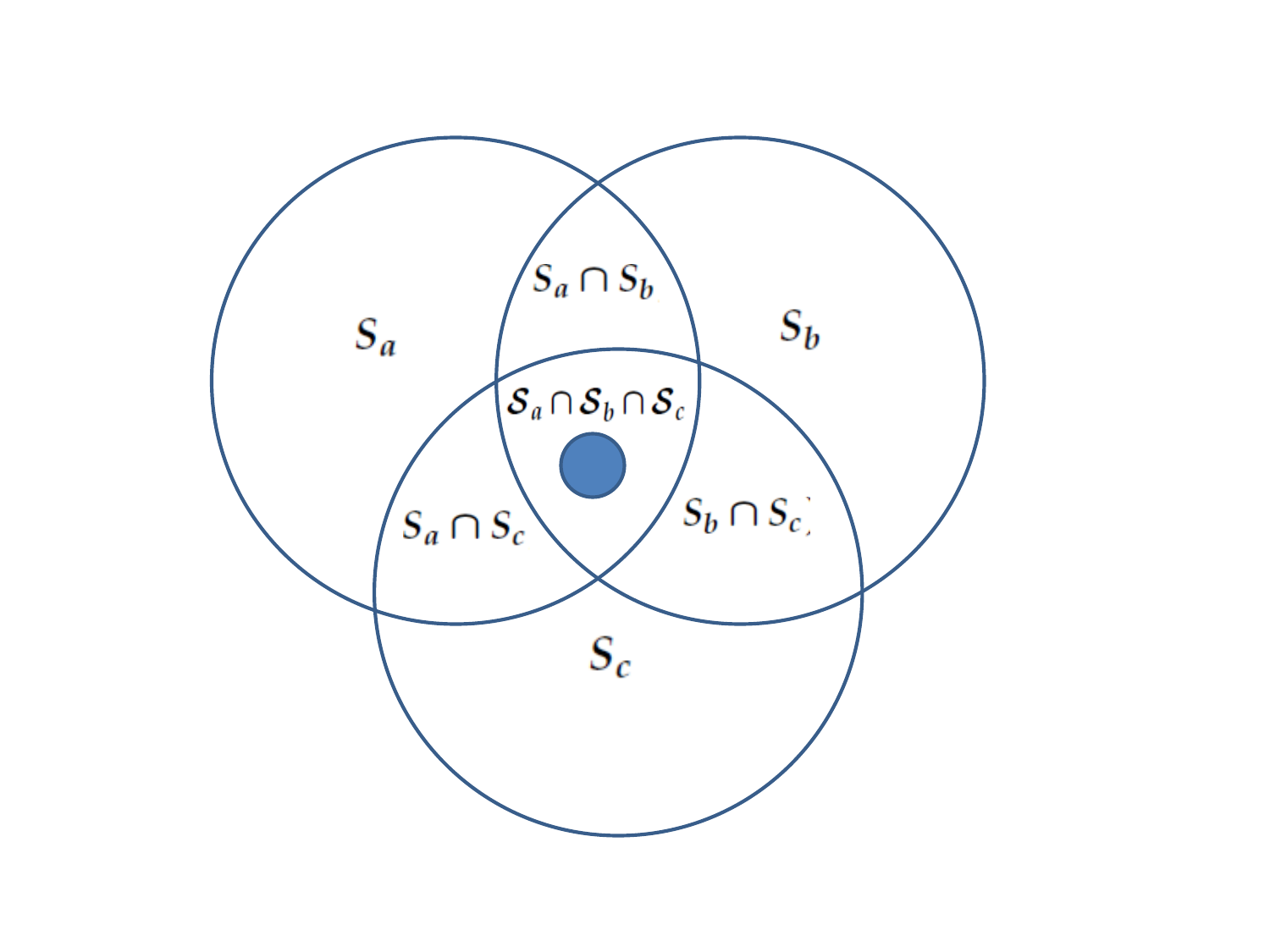}\caption{The shaded region represents the set of all input-output pairs of
sets for which the set transformation is possible using deterministic
LOCC. Note that, as discussed in the text, all regions are nonempty. }
\par\end{centering}
\end{figure}
\\

\section{Set transformations and classes of quantum operations}

Consider a bipartite system $\mathcal{H}=\mathcal{H}_{A}\otimes\mathcal{H}_{B}$.
Let $\mathcal{\bm{Q}}=\left\{ \mathcal{Q}_{1},\mathcal{Q}_{2},\dots\right\} $
be a quantum operation acting on the set of density operators $\mathcal{D}\left(\mathcal{H}\right)$.
Then for any density operator $\rho\in\mathcal{D}\left(\mathcal{H}\right)$
we have 
\begin{alignat*}{1}
\rho_{i}^{\prime} & =\mathcal{Q}_{i}\left(\rho\right).
\end{alignat*}
 Recall the definitions: 
\begin{itemize}
\item $\mathcal{\bm{Q}}\in\text{SEP}$ if each $\mathcal{Q}_{i}$ is a separable
map, which implies that $\rho_{i}^{\prime}$ is separable whenever
$\rho$ is separable. 
\item $\mathcal{\bm{Q}}\in\text{PPT}$ if each $\mathcal{Q}_{i}$ is a PPT
map, which implies that $\rho_{i}^{\prime}$ is PPT whenever $\rho$
is PPT.
\end{itemize}
The following relations hold:
\begin{equation}
\text{LOCC}\subset\text{SEP}\subset\text{PPT}\subset\text{ALL},\label{inclusion-relations}
\end{equation}
where ALL is simply the set of all quantum operations $\left\{ \mathcal{\bm{Q}}\right\} $.
The inclusions are strict. Operationally this means the following:
Let $X=\left\{ \mathcal{\bm{X}}\right\} $ and $Y=\left\{ \mathcal{\bm{Y}}\right\} $
denote two classes of quantum operations. Then $X\subset Y$ means
$Y$ is more powerful than $X$ in the sense that every $\mathcal{\bm{X}}\in X$
can be realized by some $\mathcal{\bm{Y}}\in Y$ but the converse
does not hold, i.e., not every $\mathcal{\bm{Y}}\in Y$ can be realized
by some $\mathcal{\bm{X}}\in X$. 

The purpose of this section is to show that the relations (\ref{inclusion-relations})
also hold in the context of set transformations. We begin by noting
that (\ref{inclusion-relations}) hold in the case of distinguishing
quantum states. 
\begin{lem}
\label{state-discrimination-lemma} Let $X,Y\in\left\{ \text{LOCC},\text{SEP},\text{PPT},\text{ALL}\right\} $
satisfy $X\subset Y$. Then there exists a set $S\left(X,Y\right)$
of bipartite orthogonal states that can be perfectly distinguished
with certainty by some $\mathcal{\bm{Y}}\in Y$ but not by any $\mathcal{\bm{X}}\in X.$ 
\end{lem}
The proof of Lemma \ref{state-discrimination-lemma} follows from
examples in the existing literature. Note that it suffices to cover
the cases $\text{LOCC}\subset\text{SEP}$, $\text{SEP}\subset\text{PPT}$,
and $\text{PPT}\subset\text{ALL}.$ We will give one example for each
of them. 
\begin{itemize}
\item $S\left(\text{LOCC},\text{SEP}\right)$ is an orthogonal product basis
exhibiting quantum nonlocality without entanglement \citep{NLWE}.
The basis is locally indistinguishable \citep{NLWE,Walgate-Hardy-2002}
but perfectly distinguishable using a separable measurement (obviously).
Such bases exist in all $\mathcal{H}_{A}\otimes\mathcal{H}_{B}$ with
local dimensions $d_{A},d_{B}\geq3$.
\item $S\left(\text{SEP},\text{PPT}\right)$ is the following set of three
states from $\mathcal{H}_{A}\otimes\mathcal{H}_{B}$ with local dimensions
$d_{A}=d_{B}=4$ \citep{Yu-duan-PPT-2014}:
\begin{flalign*}
\left|\psi_{1}\right\rangle  & =\frac{1}{\sqrt{2}}\left(\left|00\right\rangle +\left|11\right\rangle \right)\otimes\left(\sqrt{\frac{2}{3}}\left|00\right\rangle +\sqrt{\frac{1}{3}}\left|11\right\rangle \right)\\
\left|\psi_{2}\right\rangle  & =\frac{1}{\sqrt{2}}\left(\left|00\right\rangle -\left|11\right\rangle \right)\otimes\left(\sqrt{\frac{2}{3}}\left|00\right\rangle +\sqrt{\frac{1}{3}}\left|11\right\rangle \right)\\
\left|\psi_{3}\right\rangle  & =\frac{1}{\sqrt{2}}\left(\left|01\right\rangle +\left|10\right\rangle \right)\otimes\left(\sqrt{\frac{2}{3}}\left|00\right\rangle +\sqrt{\frac{1}{3}}\left|11\right\rangle \right)
\end{flalign*}
\item $S\left(\text{PPT},\text{ALL}\right)$ is the two-qubit Bell basis.
The proof that the Bell basis is not PPT distinguishable follows by
appropriately modifying the argument in \citep{Bell-indistinguishable}. 
\end{itemize}
We will make use of Lemma \ref{state-discrimination-lemma} to prove
the following proposition.
\begin{prop}
\label{inclusion proposition} Let $X,Y\in\left\{ \text{LOCC},\text{SEP},\text{PPT},\text{ALL}\right\} $
satisfy the relation $X\subset Y$. Then there exist input-output
pairs of sets $S_{\psi}$ and $S_{\phi}$ for which the transformation
$S_{\psi}\rightarrow S_{\phi}$ is possible with unit probability
with a $\mathcal{\bm{Y}}\in Y$ but not with any $\mathcal{\bm{X}}\in X$.
\end{prop}
\begin{proof}
Let $X,Y\in\left\{ \text{LOCC},\text{SEP},\text{PPT},\text{ALL}\right\} $
satisfy the relation $X\subset Y$. We choose the input set $S_{\psi}$
as $S\left(X,Y\right)$ and assume that its cardinality is $n$. Let
the output set $S_{\phi}$ contain any $n$ orthogonal product states
from the computational basis. Then: 
\begin{itemize}
\item The members of $S\left(X,Y\right)$ can be perfectly distinguished
with certainty by some $\mathcal{\bm{Y}}\in Y$ but not by any $\mathcal{\bm{X}}\in X$.
And, we know from Lemma \ref{state-discrimination-lemma} that such
$S\left(X,Y\right)$ exists for any choice of $X,Y\in\left\{ \text{LOCC},\text{SEP},\text{PPT},\text{ALL}\right\} $
satisfying $X\subset Y$. 
\item The members of $S_{\phi}$ are locally distinguishable, and therefore
both SEP-distinguishable and PPT-distinguishable. 
\end{itemize}
Let us first prove that the transformation $S\left(X,Y\right)\rightarrow S_{\phi}$
is achievable with certainty by some $\mathcal{\bm{Y}}\in Y$. The
protocol is simple. We know that the members of $S\left(X,Y\right)$
are distinguishable by some $\mathcal{\bm{Y}}\in Y$. So in the first
step we simply identify the input state. Since the corresponding output
state is product, it can be prepared by LOCC. Therefore, the desired
transformation is possible. 

We now show that the transformation $S\left(X,Y\right)\rightarrow S_{\phi}$
is not achievable with certainty by any $\mathcal{\bm{X}}\in X$.
Suppose, on the contrary, the transformation is, in fact, possible.
Then for every $i$, $i=1,\dots,n$, we have $\left|\psi_{i}\right\rangle \rightarrow\left|\phi_{i}\right\rangle $,
where $\left|\psi_{i}\right\rangle \in S\left(X,Y\right)$ and $\left|\phi_{i}\right\rangle \in S_{\phi}$.
Since $S_{\phi}$ is locally distinguishable, we can perform a LOCC
measurement to identify the output state, which will reveal the identity
of the input state given in the beginning. Therefore, we are able
to perfectly distinguish the members of $S\left(X,Y\right)$ with
certainty using an $\bm{\mathcal{X}}\in X$. This contradicts the
fact that $S\left(X,Y\right)$ is not perfectly distinguishable with
certainty by any $\mathcal{\bm{X}}\in X$. Hence, the proof is complete. 
\end{proof}

\section{Conclusions}

The set transformation problem in quantum information theory deals
with the question of the existence of a physical transformation that
transforms a given set of input states into a set of output states.
In this paper, we considered this problem within the LOCC framework
and discussed the unique features which arise from orthogonality,
entanglement, and local distinguishability of the states under consideration.
Specifically, we investigated the problem of transforming a set of
pure bipartite states into another using deterministic LOCC. We obtained
the necessary conditions for the existence of such a transformation
using LOCC constraints on state transformation, entanglement, and
distinguishability. We proved the conditions are independent but not
sufficient. We discussed their satisfiability and classified all possible
input-output pairs of sets accordingly. We also showed the strict
inclusion relations that hold for LOCC, separable, and PPT operations
apply in the case of set transformations. 

We, however, did not attempt to address the most general LOCC set
transformation problem. In particular, we left out mixed states and
multipartite systems. Not much is known about local mixed-state transformations
and local distinguishability, and entanglement of mixed states is
also extremely hard to compute. In fact, the main tools we used in
this paper will not work that well for mixed states. These are the
reasons why we did not consider mixed states.

Multipartite systems also pose considerable challenges, as neither
state transformation, entanglement, nor distinguishability properties
are well understood, at least not as well as bipartite systems. However,
we believe that many of the results presented in this paper could
be extended to problems involving multipartite pure states without
much difficulty. This scenario could be an exciting avenue for further
research. 
\begin{acknowledgement*}
R.S. acknowledges financial support from DST Project No. DST/ICPS/QuST/
Theme-2/2019/General Project Q-90. S.H. was supported by a postdoctoral
fellowship of Harish-Chandra Research Institute, Prayagraj (Allahabad)
when part of this work was completed. S.H. is now supported by the
\textquotedblleft Quantum Optical Technologies\textquotedblright{}
project, carried out within the International Research Agendas programme
of the Foundation for Polish Science co-financed by the European Union
under the European Regional Development Fund.
\end{acknowledgement*}


\begin{thebibliography}{99}
\bibitem{Chitambar+LOCC} E. Chitambar, D. Leung, L. Man\v{c}inska,
M. Ozols, and A. Winter, Everything You Always Wanted to Know About
LOCC (But Were Afraid to Ask) \href{https://link.springer.com/article/10.1007/s00220-014-1953-9}{Comm. Math. Phys. {\bf328}, 303326 (2014)}.

\bibitem{Bell-nonlocality-review} N. Brunner, D. Cavalcanti, S. Pironio,
V. Scarani, and S. Wehner, Bell nonlocality, \href{https://doi.org/10.1103/RevModPhys.86.419}{Rev. Mod. Phys. {\bf 86}, 419 (2014)};
Erratum, \href{https://doi.org/10.1103/RevModPhys.86.839}{Rev. Mod. Phys. {\bf 86}, 839 (2014)}.

\bibitem{NLWE} C. H. Bennett, D. P. DiVincenzo, C. A. Fuchs, T. Mor,
E. Rains, P. W. Shor, J. A. Smolin, and W. K. Wootters, Quantum nonlocality
without entanglement, \href{https://doi.org/10.1103/PhysRevA.59.1070}{Phys. Rev. A {\bf59}, 1070 (1999)}.

\bibitem{ben99u} C. H. Bennett, D. P. DiVincenzo, T. Mor, P. W. Shor,
J. A. Smolin, and B. M. Terhal, Unextendible Product Bases and Bound
Entanglement, \href{https://doi.org/10.1103/PhysRevLett.82.5385}{Phys. Rev. Lett. {\bf82}, 5385 (1999)}.

\bibitem{Horodeckis+MNLE} M. Horodecki, A. Sen(De), U. Sen, and K.
Horodecki, Local Indistinguishability: More Nonlocality with Less
Entanglement, \href{https://doi.org/10.1103/PhysRevLett.90.047902}{Phys. Rev. Lett. {\bf90}, 047902 (2003)}.

\bibitem{Bandyo-2011} S. Bandyopadhyay, More Nonlocality with Less
Purity, \href{https://doi.org/10.1103/PhysRevLett.106.210402}{Phys. Rev. Lett. {\bf106}, 210402 (2011)}.

\bibitem{Halder+-2019} S. Halder, M. Banik, S. Agrawal, and S. Bandyopadhyay,
Strong Quantum Nonlocality without Entanglement,\textbf{ }\href{https://doi.org/10.1103/PhysRevLett.122.040403}{Phys. Rev. Lett. {\bf122}, 040403 (2019)}. 

\bibitem{resource-theories-review} E. Chitambar and G. Gour, Quantum
resource theories, \href{https://journals.aps.org/rmp/abstract/10.1103/RevModPhys.91.025001}{Rev. Mod. Phys. {\bf91}, 025001  (2019)}.

\bibitem{Peres-Wootters} A. Peres and W. K. Wootters, Optimal Detection
of Quantum Information, \href{https://doi.org/10.1103/PhysRevLett.66.1119}{Phys. Rev. Lett. {\bf66}, 1119 (1991)}.

\bibitem{Walgate+2000} J. Walgate, A. J. Short, L. Hardy, and V.
Vedral, Local Distinguishability of Multipartite Orthogonal Quantum
States, \href{https://doi.org/10.1103/PhysRevLett.85.4972}{Phys. Rev. Lett. {\bf85}, 4972 (2000)}.

\bibitem{Bell-indistinguishable} S. Ghosh, G. Kar, A. Roy, A. Sen(De),
and U. Sen, Distinguishability of Bell States, \href{https://journals.aps.org/prl/abstract/10.1103/PhysRevLett.87.277902}{Phys. Rev. Lett. {\bf87}, 277902 (2001)}. 

\bibitem{Walgate-Hardy-2002} J. Walgate and L. Hardy, Nonlocality,
Asymmetry, and Distinguishing bipartite states, \href{https://doi.org/10.1103/PhysRevLett.89.147901}{Phys. Rev. Lett. {\bf89}, 147901 (2002)}.

\bibitem{Watrous-2005} J. Watrous, Bipartite Subspaces Having No
Bases Distinguishable by Local Operations and Classical Communication,
\href{https://doi.org/10.1103/PhysRevLett.95.080505}{Phys. Rev. Lett. {\bf95}, 080505 (2005)}.

\bibitem{Nathanson-2005} M. Nathanson, Distinguishing bipartite orthogonal
states by LOCC: best and worst cases, \href{https://doi.org/10.1063/1.1914731}{Journal of Mathematical Physics {\bf46}, 062103 (2005)}.

\bibitem{BGK-2011} S. Bandyopadhyay, S. Ghosh and G. Kar, LOCC distinguishability
of unilaterally transformable quantum states, \href{https://doi.org/10.1088/1367-2630/13/12/123013}{New J. Phys. {\bf13}, 123013 (2011)}.

\bibitem{Yu-Duan-2012} N. Yu, R. Duan, and M. Ying, Four Locally
Indistinguishable Ququad-Ququad Orthogonal Maximally Entangled States,
\href{https://doi.org/10.1103/PhysRevLett.109.020506}{Phys. Rev. Lett. {\bf109}, 020506 (2012)}.

\bibitem{Cosentino-2013} A. Cosentino, Positive-partial-transpose-indistinguishable
states via semidefinite programming, \href{https://doi.org/10.1103/PhysRevA.87.012321}{Phys. Rev. A {\bf87}, 012321 (2013)}.

\bibitem{BN-2013} S. Bandyopadhyay, M. Nathanson, Tight bounds on
the distinguishability of quantum states under separable measurements,
\href{https://doi.org/10.1103/PhysRevA.88.052313}{Phys. Rev. A {\bf88} 052313 (2013)}.

\bibitem{Cosentino-Russo-2014} A. Cosentino and V. Russo, Small sets
of locally indistinguishable orthogonal maximally entangled states,
\href{https://doi.org/10.26421/QIC14.13-14}{Quantum Information and Computation {\bf14}, 1098 (2014)}.

\bibitem{B-IQC-2015} S. Bandyopadhyay, A. Cosentino, N. Johnston,
V. Russo, J. Watrous, and N. Yu, Limitations on separable measurements
by convex optimization, \href{https://doi.org/10.1109/TIT.2015.2417755}{IEEE Transactions on Information Theory, {\bf61},  3593 (2015)}.. 

\bibitem{entanglement-review} R. Horodecki, P. Horodecki, M. Horodecki,
and K. Horodecki, Quantum entanglement, \href{https://journals.aps.org/rmp/abstract/10.1103/RevModPhys.81.865}{Rev. Mod. Phys. {\bf81}, 865  (2009)}.

\bibitem{Nielsen-99} M. A. Nielsen, Conditions for a Class of Entanglement
Transformations, \href{https://journals.aps.org/prl/abstract/10.1103/PhysRevLett.83.436}{Phys. Rev. Lett. {\bf83}, 436 (1999)}.

\bibitem{Alberti-Uhlman-1980-1} P. Alberti and A. Uhlmann, Existence
and Density Theorems for Stochastic Maps on Commutative $C^{*}$-Algebras,
\href{https://onlinelibrary.wiley.com/doi/abs/10.1002/mana.19800970125}{Math. Nachr. {\bf97}, 279 (1980)}. 

\bibitem{AlbertiUhlman-1980-2} P. Alberti and A. Uhlmann, A problem
relating to positive linear maps on matrix algebras, \href{https://www.sciencedirect.com/science/article/pii/003448778090083X}{Rep. Math. Phys. {\bf18}, 163 (1980)}.

\bibitem{Uhlmann-1985} A. Uhlmann, The transition probability for
states of \textasteriskcentered -algebras, \href{http://www.physik.uni-leipzig.de/~uhlmann/PDF/Uh85a.pdf}{Annalen der Physik {\bf42}, 524 (1985)}.

\bibitem{Chefles-2000} A. Chefles, Deterministic quantum state transformations,
\href{https://www.sciencedirect.com/science/article/pii/S0375960100002917}{Phys. Lett. A {\bf270}, 14 (2000)}.

\bibitem{Chefles+} A. Chefles, R. Jozsa, and A. Winter, On the existence
of physical transformations between sets of quantum states, \href{https://www.worldscientific.com/doi/abs/10.1142/S0219749904000031}{Int. J. Quantum Inf.  {\bf02}, 11 (2004)}. 

\bibitem{BH-2021} S. Bandyopadhyay and S. Halder, Genuine activation
of nonlocality: From locally available to locally hidden information,
\href{https://journals.aps.org/pra/pdf/10.1103/PhysRevA.104.L050201}{Phys. Rev. A {\bf104}, L050201 (2021)}.

\bibitem{Jozsa+1999} R. Jozsa and J. Schlienz, Distinguishability
of states and von Neumann entropy, \href{https://journals.aps.org/pra/abstract/10.1103/PhysRevA.62.012301}{Phys. Rev. A {\bf62}, 012301 (2000)}.

\bibitem{Vidal-1999} G. Vidal, Entanglement of Pure States for a
Single Copy,\href{https://journals.aps.org/prl/abstract/10.1103/PhysRevLett.83.1046}{Phys. Rev. Lett. {\bf83}, 1046 (1999)}.

\bibitem{Navascues-2008} M. Navascués, Pure State Estimation and
the Characterization of Entanglement,\href{https://journals.aps.org/prl/abstract/10.1103/PhysRevLett.100.070503}{Phys. Rev. Lett. {\bf100}, 070503 (2008)}.

\bibitem{Fuchs-2004} C. A. Fuchs, On the quantumness of a Hilbert
space, \href{https://dl.acm.org/doi/abs/10.5555/2011593.2011599}{Quant. Inf. and Comp. {\bf4} 467 (2004)}.

\bibitem{Partial=00003Dtranspose} M. Horodecki, P. Horodecki, and
R. Horodecki, On the necessary and sufficient conditions for separability
of mixed quantum states, \href{http://citeseerx.ist.psu.edu/viewdoc/download?doi=10.1.1.251.8530&rep=rep1&type=pdf}{Phys. Lett. A {\bf223}, 11 (1996)}. 

\bibitem{Concurrence} William K. Wootters, Entanglement of Formation
of an Arbitrary State of Two Qubits, \href{https://journals.aps.org/prl/abstract/10.1103/PhysRevLett.80.2245}{Phys. Rev. Lett. {\bf80}, 2245  (1998)}.

\bibitem{Yu-duan-PPT-2014} N. Yu, R. Duan, and M. Ying, Distinguishability
of quantum states by positive operator-valued measures with positive
partial transpose, \href{https://ieeexplore.ieee.org/abstract/document/6747300}{IEEE Trans. Inf. Theory {\bf 60}, 2069 (2014)}.
\end{thebibliography}
\end{document}